\documentclass{llncs}
\usepackage{makeidx}  % allows for indexgeneration
\usepackage{bbm,amsmath,amsfonts}
\usepackage[dvips]{graphicx}
\newcommand\ind[1]{\mathbbm{1}_{\{#1\}}}
\def\N{{\mathbb N}}
\def\R{{\mathbb R}}
\def\E{{\mathbb E}}
\def\etal{{\em et al.}}
\def\cal{\mathcal}
\graphicspath{{Fig/}}
\begin{document}
\frontmatter          % for the preliminaries
\pagestyle{headings}  % switches on printing of running heads
\mainmatter              % start of the contributions
\title{Stability Properties of Networks with Interacting TCP Flows}
\titlerunning{Interacting TCP Flows}  % abbreviated title (for runningA_{jk} z_k head)

\author{Carl Graham\inst{1} \and Philippe Robert\inst{2} \and Maaike
  Verloop\inst{3}\thanks{Part of this work was done during a 3-month visit of Maaike
    Verloop at INRIA Paris --- Rocquencourt with financial support of the European Network
    of Excellence EURO-NF.}} 
\authorrunning{Graham, Robert and Verloop}   % abbreviated author list (for running head)
\tocauthor{Carl Graham, Philippe Robert and Maaike Verloop}
\institute{UMR 7641  CNRS --- \'Ecole Polytechnique, Route de Saclay, 91128 Palaiseau, France\\
\email{carl@cmapx.polytechnique.fr}
\and
INRIA Paris --- Rocquencourt, Domaine de Voluceau, 78153 Le Chesnay, France\\
\email{Philippe.Robert@inria.fr}\\
\texttt{http://www-rocq.inria.fr/\~{}robert}
\and
CWI, P.O. Box 94079, 1090 GB Amsterdam, The Netherlands\\
\email{I.M.Verloop@cwi.nl}\\
\texttt{http://www.cwi.nl/\~{}maaike}
}

\maketitle              % typeset the title of the contribution

\begin{abstract}
The equilibrium distributions  of a Markovian model describing  the interaction of several
classes of  permanent connections  in a network  are analyzed.  It has been  introduced by
Graham  and  Robert~\cite{Graham:01}.   For this  model  each  of  the connections  has  a
self-adaptive behavior in that its transmission  rate along its route depends on the level
of  congestion  of   the  nodes  on its   route.   It  has  been  shown
in~\cite{Graham:01} that  the invariant distributions are determined  by the solutions 
of a fixed point equation in a finite dimensional space. In this paper, several examples
of these fixed point equations are studied. The topologies investigated are rings, trees and a linear network, with
various sets of routes through the nodes. 
\end{abstract}

\section{Introduction}
Data transmission in  the Internet network can be described as a  self-adaptive system to the
different congestion  events that regularly occur  at its numerous nodes.  A connection, a
TCP flow, in this network adapts  its throughput according to the congestion it encounters
on its path: Packets are sent as long as no loss is detected and throughput grows linearly
during that time.  On  the contrary when a loss occurs, the  throughput is sharply reduced
by  a multiplicative  factor.   This scheme  is known  as an  Additive  Increase and
Multiplicative Decrease algorithm (AIMD).

Globally,  the TCP  protocol  can  be seen  as  a bandwidth  allocation  algorithm on  the
Internet.  From a  mathematical modelling  perspective, the  description is  somewhat more
difficult. While the representation  of the evolution of the throughput of  a single TCP flow
has  been the  object  of  various rigorous  works,  there are  few  rigorous studies  for
modelling the evolution of a large set of TCP connections in a quite large network.

A possible  mathematical formulation which has  been used is via  an optimization problem:
given  $K$   classes  of   connections,  when  there   are  $x_k$  connections   of  class
$k\in\{1,\ldots,K\}$, their total throughput achieved  is given by $\lambda_k$ so that the
vector $(\lambda_k)$ is a solution of the following optimization problem
\[
\max_{\lambda\in\Lambda} \sum_{k=1}^K x_kU_k(\lambda_k/x_k),
\]
where $\Lambda$ is the set of admissible throughputs which takes into account the capacity
constraints of the network. The functions  $(U_k)$ are defined as {\em utility} functions, and 
various expressions have been proposed for them. See Kelly \etal~\cite{Kelly:02},
Massouli\'e~\cite{Massoulie:01} and Massouli\'e   and    Roberts~\cite{Massoulie}. With this
representation, the TCP protocol is seen as an adaptive algorithm maximizing some
criterion at the level of the network. 

A different point of view has been proposed in Graham and Robert~\cite{Graham:01}. It
starts on the local dynamics of the AIMD algorithm used by TCP and, through a scaling
procedure, the global behavior of the network can then be described rigorously. It is assumed
that there are $K$ classes of {\em permanent } connections going through different nodes
and with different characteristics. The loss rate of a connection using a given node $j$,
$1\leq j\leq J$, is described as a function of the  congestion $u_j$ at this node. The
quantity $u_j$ is defined as the (possibly weighted) sum of the throughputs of {\em all}
the connections that use node $j$.  The interaction of the connections in the network is
therefore expressed via the loss rate at each node. 

It has been  shown in Graham and Robert~\cite{Graham:01} that  under a mean-field scaling,
the evolution of a class $k$  connection, $1\leq k\leq K$, can be asymptotically described
as the unique solution of an unusual stochastic differential equation. Furthermore, it has
also been  proved that the  equilibrium distribution of  the throughputs of  the different
classes of  connections is in  a one to  one correspondence with  the solution of  a fixed
point  equation $({\cal E})$ of  dimension  $J$  (the  number  of  nodes).  

Under ``reasonable'' conditions, there should  be only one solution of $({\cal E})$ and
consequently  a unique stable equilibrium of the network.  Otherwise this would 
imply that the state of the network could oscillate between several stable states. Although this is  mentioned here and there in the literature, this has not been firmly established in the context  of an  IP network.  It has  been  shown that  multi-stability may  occur in  loss
networks, see Gibbens \etal~\cite{Gibbens:02} and Marbukh~\cite{Marbukh:01} or in the context
of a  wireless network with admission  control, see Antunes  \etal~\cite{Antunes:06}. Raghunathan
and Kumar~\cite{Rag} presents  experiments that suggest that a  phenomenon of bi-stability
may occur in a
context similar to the one considered  in this paper but for wireless networks.

It turns out that it is  not easy to check in practice whether the fixed point
equation $({\cal E})$ has  a unique solution  or not. The  purpose of this  paper is to  investigate in
detail this question for several topologies. The paper is organized as
follows. Section~\ref{secOne} reviews the main definitions and results used in the paper. 
In addition, a simple criterion for the existence of a fixed-point solution is given. 
Section~\ref{secTree} presents a uniqueness result for a tree topology under the assumptions that all 
connections use the root. Section~\ref{secLinear} considers a linear network. Section~\ref{Torus} studies several scenarios for ring topologies and a uniqueness result is proved for  connections going through one, two, or 
all the nodes. 
Two main approaches are used to prove uniqueness: monotonicity properties
of the network and contraction arguments.  

A general conjecture that we make is that when the loss rates are increasing in the level of congestion, this should be
sufficient to imply the uniqueness of the equilibrium in a general network (together with regularity properties perhaps). 

\section{A Stochastic Fluid Picture}\label{secOne}
\setcounter{equation}{0}

In this section, a somewhat simplified version of the stochastic model of interacting TCP flows
of Graham and Robert~\cite{Graham:01} is presented. 

\subsection*{The case of a single connection}

Ott  \etal~\cite{Ott:01} presents a fluid model of a single connection.  Via  scalings  with
respect to  the  loss rate,  Dumas \etal~\cite{Dumas:06}  proves various limit  theorems
for the resulting processes. The limiting picture of Dumas \etal~\cite{Dumas:06} for the
evolution of the throughput of single long connection is as follows. 

If the instantaneous throughput at time $t$ of the connection is 
$W(t)$, this process has the Markov property and its infinitesimal generator is given by   
\begin{equation}\label{gen}
\Omega(f)(x)=a f'(x)+ \beta x(f(rx)-f(x))
\end{equation}
for $f$ a  $C^1$-function from $\R_+$ to  $\R$. For $t\geq 0$, the quantity $W(t)$ should
be thought as the instantaneous throughput of the connection at time $t$.  

The Markov process $(W(t))$ increases linearly at rate $a$. 
The  constant $a$ is related to the distance between the source and the
destination. It increases proportionally to the round trip time $RTT$, typically
\[
a=\frac{C_0}{C_1+RTT},
\]
for some constants $C_0$ and $C_1$.

Given $W(t)=x$, the process $(W(t))$ jumps from $x$ to $rx$ ($r$ is usually $1/2$) at rate $\beta x$. The expression $\beta x$ represents the loss  rate of the connection. %The constant $\beta$ depends on the level of the buffer of the node, i.e. on its congestion. 
Of course, the quantities $a$,  $\beta$ and
$r$  depend on the parameters of the connection. 

The density of the invariant distribution of this Markov process is given in the following
proposition.  It  has been analyzed in Ott  \etal~\cite{Ott:01} at the fluid  level and by
Dumas  \etal~\cite{Dumas:06},   see  also  Guillemin \etal~\cite{Guillemin:05}.   The transient
behavior has been investigated in Chafai \etal~\cite{Chafai}.
\begin{proposition}\label{OneProp}
The function
\begin{equation}\label{densityH}
H_{r,\rho}(w)=\frac{\sqrt{2\rho/\pi}}{\prod_{n=0}^{+\infty}(1-r^{2n+1})}
\sum_{n=0}^{+\infty} \frac{r^{-2n}}{\prod_{k=1}^{n}(1-r^{-2k})}
e^{-\rho r^{-2n}{w^2}/{2}}, \quad w\geq 0,
\end{equation}
with $\rho=a/\beta$, is the density of the invariant distribution of the Markov process
$(W(t))$ whose infinitesimal generator is given by Equation~\eqref{gen}. Furthermore, its
expected value is given by 
\begin{equation}\label{Einv}
\int_0^{+\infty}wH_{r,\rho}(w)\,dw= \sqrt{\frac{2\rho}{\pi}} \prod_{n=1}^{+\infty} \frac{1-r^{2n}}{1-r^{2n-1}}.
\end{equation}
\end{proposition}

\subsection*{A Representation of Interacting Connections in a Network}
The    network    has   $J\ge1$    nodes    and    accommodates    $K\ge1$   classes    of
permanent connections. For $1\leq k\leq K$,  the number of class $k$ connections is
$N_k\geq 1$, and one sets
\[
N=(N_1,\ldots,N_K), \text{ and }\quad |N| = N_1 + \cdots + N_K.
\]
An \emph{allocation  matrix} $A = (A_{jk}, 1  \le j \le J,  1 \le k \le  K)$ with positive
coefficients describes the use of nodes by the connections. In particular the route of a
class~$k$ connection goes through node $j$ only if $A_{jk}>0$. In practice, the class of a
connection is determined by the sequence set of nodes it is using.

If $w_{n,k}\geq 0 $ is the  throughput of the  $n$th class  $k$ connection,  $1\leq n\leq
N_k$, the  quantity $A_{jk} w_{n,k}$ is the weighted throughput at node $j$ of this
connection. A  simple example would be to take $A_{jk}=1$ or $0$  depending on whether 
 a  class~$k$ connection uses  node $j$  or not.   The total weighted throughput $u_j$
of node $j$ by the various connections is given by 
\[
u_j = \sum_{k=1}^K\sum_{n=1}^{N_k}  A_{jk} w_{n,k}.
\]
The quantity $u_j$ represents the level of utilization/congestion of node $j$. In particular, the loss
rate  of a  connection going  through node~$j$  will depend  on this  variable.  

For $1\leq k\leq K$, the corresponding  parameters $a$ and $\beta$ of Equation~\eqref{gen} for
a class $k$ connection are given by  a non-negative number $a_k$
and a function $\beta_k :  \R_+^J \rightarrow \R_+$, so that when  the resource vector of the
network is  $u = (u_j, {1  \le j \le J})$  and if the state  of a class  $k$ connection is
$w_k$:
\begin{itemize}
\item  Its state increases linearly at rate $a_k$. For example $a_k=1/R_k$ where
$R_k$ is the round trip time between the source and the destination of a class $k$
connection. 
\item  A loss for this connection occurs at rate  $w_k\beta_k(u)$ and in this case its
  state jumps from $w_k$ to $r_kw_k$.  The function $\beta_k$ depends only on the utilization of all
  nodes used by class $k$  connections. In particular, if a class $k$ connection goes
  through the nodes ${j_1}$, ${j_2}$, \ldots, ${j_{l_k}}$, one has 
\[
\beta_k(u)=\beta_k(u_{j_1},u_{j_2},\ldots,u_{j_{l_k}}).
\]
A more specific  (and natural) choice for $\beta_k$ would be 
\begin{equation}
\label{eq:delta}
\beta_k(u)=\delta_k + \varphi_{j_1}(u_{j_1})+\varphi_{j_2}(u_{j_2})+\cdots+\varphi_{j_{J_k}}(u_{j_{l_k}}),
\end{equation}
where $\varphi_{j_\ell}(x)$ is the loss rate at node $j_\ell$ when its congestion level is
$x\geq 0$, and $\delta_k$ is the loss rate in a non-congested network.
Another example is when the loss rate $\beta_k$ depends only on the sum of the utilizations of the nodes used by class~$k$, i.e., 
\begin{equation}
 \label{eq:lossr}
\beta_k(u)= \beta_k\left(\sum_{l=1}^{l_k} u_{j_{l}}\right).
\end{equation}
\end{itemize}

\subsection*{Asymptotic behaviour of typical connections}
If $(W_{n,k}(t))$ denotes the  throughput of the  $n$th class  $k$
connection,  $1\leq n\leq N_k$, then the vector 
\[
(W(t))=([(W_{n,k}(t)), 1\leq k\leq K, 1\leq n\leq N_k], t\geq 0)
\]
has the Markov property. As it stands, this Markov process is quite difficult to
analyze. For this reason, a mean field scaling is used to get a more quantitative
representation of the interaction of the flows. 
More specifically, it is assumed  that the total number of connections $\|N\|$ 
goes to infinity and that the total number of class $k$ connections is of the order 
$p_k\|N\|$, where $p_1+\cdots+p_K=1$. 

For each $1\leq k\leq K$, one takes a class $k$ connection  at random, let  $n_k$ be its
index,  $1\leq n_k\leq N_k$. The process $(W_{n_k,k}(t))$ represents the throughput of a ``typical''
class $k$  connection. It is shown in Graham and Robert~\cite{Graham:01} that, as $\|N\|$
goes to infinity and under mild assumptions, the process  $[(W_{n_k,k}(t)),
  1\leq k\leq K]$  converges in distribution to 
$(\overline{W}(t))=[(\overline{W}_{k}(t)), 1\leq k\leq K]$,
where the processes $(\overline{W}_{k}(t))$, for $1\leq k \leq K$, are
independent and, for $1\leq k\leq K$, the process $(\overline{W}_{k}(t))$ is the solution of
the following stochastic differential equation,
\begin{equation}\label{nlsde}
d\overline{W}_k(t) = a_k\, dt 
- ( 1-r_k) \overline{W}_k(t-) 
\int \ind{0 \le z \le  \overline{W}_k(t-)\beta_k\left(u_{\overline{W}}(t)\right)}
\,{\cal N}_k(dz, dt),
\end{equation}
with ${u}_{\overline{W}}(t) = (u_{{\overline{W},j}}(t), 1 \le j \le J)$ and, for $1 \le j \le J$,
\[
u_{\overline{W},j}(t) = \sum_{k=1}^K A_{jk} p_k \E(\overline{W}_k(t)),
\]
where $({\cal  N}_k, 1\leq k\leq K)$ are  i.i.d.\ Poisson point processes  on $\R_+^2$ with
Lebesgue characteristic measure. 

Because of the role of the  deterministic function $(u_{\overline{W}(t)})$ in these
equations, the Markov  property  holds for this process but it is not  {\em
  time-homogeneous}. The analogue of the infinitesimal generator 
$\overline{\Omega}_{k,t}$ is given by 
\[
\overline{\Omega}_{k,t}(f)(x)=a_kf'(x)+ x\beta_k ({u}_{\overline{W}}(t))(f(r_kx)-f(r_k)).
\]
The homogeneity holds when the function  $({u}_{\overline{W}}(t))$ is equal to a constant 
$u^*$, which will be the case at equilibrium. In this case a class $k$ connection behaves
like a single isolated connection with parameters $a=a_k$ and $\beta=\beta_k(u^*)$.

\subsection*{The Fixed Point Equations}
The following theorem gives a characterization of the invariant distributions for
the process $(\overline{W}(t))$. 
\begin{theorem}\label{thinv}
The invariant distributions for solutions $(\overline{W}(t))$ of Equation~\eqref{nlsde} 
are in one-to-one correspondence with the solutions $u\in\R_+^J$ of the fixed point equation   
\begin{equation}\label{fp}
u_j=\sum_{k=1}^K A_{jk} \phi_k(u), \quad 1\leq j\leq J,
\end{equation}
where  
\begin{equation}\label{Irancy}
\phi_k(u)=p_k\sqrt{\frac{2}{\pi}}\left(\prod_{n=1}^{+\infty} \frac{1-r_k^{2n}}{1-r_k^{2n-1}}\right)\;\sqrt{\frac{a_k}{\beta_k(u)}}.
\end{equation}
If $u^*$ is such a solution, the corresponding invariant distribution has the density
$w\to \prod_{k=1}^K H_{r_k,\rho_k}(w_k)$ on $\R_+^K$, where $\rho_k=a_k/\beta_k(u^*)$ and $H_{r,\rho}$ is
defined in Proposition~\ref{OneProp}.
\end{theorem}
The above theorem shows that if  the fixed point equation~\eqref{fp} has several solutions, 
then the  limiting process $(\overline{W}(t))$ has  several invariant distributions. Similarly, if
equation~\eqref{fp} has  no solution  then, in particular, $(\overline{W}(t))$ cannot converge  to an
equilibrium. These  possibilities have been  suggested in the Internet  literature through
simulations, like the cyclic behavior of  some nodes in the case of congestion. 

Under  mild and  natural assumptions,  such  as the  loss rate  being non-\linebreak decreasing  with
respect to the utilization  in the nodes, we show that for some  specific topologies  there
exists a unique fixed point. We believe  that such a uniqueness result will hold, in fact,
for  any  network  in general  (under  suitable  regularity  properties on  the  functions
$\beta_k$,  $k=1,\ldots, K$).   Before proceeding  to the  examples, we  first  present an
existence result that holds for a general network.

\subsection{An Existence Result}
In this section the existence of a solution to the fixed-point equation~(\ref{fp})  is proved for a quite general framework.

If $u$ is a solution of Equation~\eqref{fp} and $z_k=\phi_k(u)$, $1\leq k\leq K$, then the vector $z=(z_k)$ satisfies the relation
$u=Az$, i.e., $u_j=A_{j1} z_1+A_{j2} z_2+\cdots+A_{jK} z_K$, $1\leq j\leq J$, and as well  
\begin{equation}\label{fp1}
z=\Phi(z)\stackrel{\text{def.}}{=}(\phi_k(A z), 1\leq k\leq K).
\end{equation}The proposition below gives a simple criterion for the existence of a fixed point. 
\begin{proposition}
\label{prop:fp_exist}
If the functions $u\to \beta_k(u)$, $1\leq k \leq K$, are  continuous and non-decreasing,  and if
there exists a vector $z^{(0)}\in \mathbb{R}^K_+$ such that the relations 
\[
z^{(0)}\leq \Phi(z^{(0)}), z^{(0)}\leq \Phi(\Phi(z^{(0)})), \text{ and } \Phi(z^{(0)})<\infty,
\]
hold coordinate by coordinate,  then 
there exists at least one solution for the fixed point Equation~\eqref{fp1} and therefore also for
Equation~\eqref{fp}.
\end{proposition}
\begin{proof}
Define   the  sequence   $z^{(n)}=\Phi(z^{(n-1)})$,  $n=1,2,\ldots$.    From  $z^{(0)}\leq
\Phi(z^{(0)})$  and  $z^{(0)}\leq   \Phi(\Phi(z^{(0)}))$,  it  follows  that  $z^{(0)}\leq
z^{(1)}$ and $z^{(0)}\leq z^{(2)}$. Since the function $\Phi$ is non-increasing, one gets
that  the relation
\[
z^{(0)}\leq z^{(2)}\leq \ldots \leq z^{(2n)} \leq \ldots \leq z^{(2n+1)} \leq \ldots \leq z^{(3)} \leq z^{(1)}
\]
holds.
Hence, there are $z_*, z^*\in \mathbb{R}_+^K$ such that 
\[
\lim_{n\to \infty} z^{(2n)} = z_* \text{ and } \lim_{n\to \infty}  z^{(2n-1)}=z^*,
\]
with $z_*\leq z^*$. Since $z^{(2n)}=\Phi(z^{(2n-1)})$ and $z^{(2n+1)}=\Phi(z^{(2n)})$, by
continuity we also have that $z^*=\Phi(z_*)$  and $z_*=\Phi(z^*)$. 

Define the set $D=\{z: z_*\leq z\leq z^*\}$. Note that $z^*\leq z^{(1)}=\Phi(z^{(0)})<\infty$, hence $D$ is bounded. In addition,  for $z\in D$,  
\[
z_*=\Phi(z^*)\leq \Phi(z) \leq \Phi(z_*)=z^*,
\] 
since the function $\Phi$ is non-increasing. One can therefore apply Brouwer fixed point theorem to $\Phi$
restricted to the compact convex set $D$, and conclude that $D$ contains at least one
fixed point of the function $\Phi$. The proposition is proved.
\end{proof}
The conditions of Proposition~\ref{prop:fp_exist} trivially hold when $\beta_k$ is non-decreasing and  $\beta_k(0)>0$
for all $k=1,\ldots, K$, since then $\Phi(0)<\infty$,  $0\leq \Phi({0})$ and ${0}\leq
\Phi(\Phi({0}))$. In particular, when the function $\beta_k$ is given by~(\ref{eq:delta}),  $\delta_k>0$ is a
sufficient condition for the existence of a fixed point.

\section{Tree topologies}\label{secTree}

We consider a finite tree network. A connection starts in the root
and then follows  the tree structure until it  leaves the network at some  node. The set
of routes is therefore indexed by the set of nodes, i.e.,  a connection following route~$G\in{\cal T}$ starts in the root  and leaves the tree in node $G$. 

The tree can be  classically represented  as  a subset ${\cal  T}$ of $\cup_{n\geq 0}\N^n$
with the constraint that if  $G=(g_1,\ldots,g_p)\in{\cal  T}$, then, for  $1\leq  \ell\leq
p$,  the element $H=(g_1,\ldots, g_\ell)$ is a node of  the tree as well. In addition,
node~$H$ is the $g_\ell$th child of generation (level) $\ell$ and the ancestor of $G$ for
this generation. One writes $H\subseteq G$ in this situation and $H\vdash G$ when
$\ell=p-1$, i.e., when $G$ is a daughter of $H$. The quantity $u_{[H,G]}$  denotes the
vector $(u_P,P: H\subseteq P\subseteq G)$. The root of the tree  is denoted by $\emptyset$. 
\begin{figure}[ht]
\begin{center}
\scalebox{.3}{\includegraphics{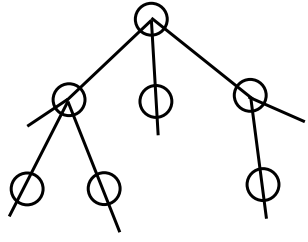}}
\end{center}
\caption{Tree with connections starting at root node}
\end{figure}
Assume that $A_{HG}=1$ if route~$G$ uses node~$H$, and 0 otherwise. 
Equation~\eqref{fp} writes in this case,   
\begin{equation*}
u_H=\sum_{G\in {\cal T},  H\subseteq G }  \phi_{G}(u_{[\emptyset,G]}),  \ \ \ H\in{\cal T},
\end{equation*}
which is equivalent to the recursive equations
\begin{equation}\label{eq:leave}
 u_H=\phi_H(u_{[\emptyset, H]}) + \sum_{G\in {\cal T},  H \vdash G} u_{G}, \ \ \ H\in{\cal T}.
\end{equation}
%%%

\begin{proposition}
If the functions $\beta_H, H\in {\cal  T}$, are continuous and non-decreasing, then there exists a unique solution for the fixed point equation~(\ref{fp}).
\end{proposition}
\begin{proof}
Let $H$  be a  maximal element on  ${\cal T}$ for  the relation  $\subseteq $, i.e., $H$ is a
leaf, and denote by $P(H)$ the parent of node~$H$. Equation~(\ref{eq:leave}) then writes
\begin{equation}
\label{eq:leave2}
u_H=\phi_H(u_{[\emptyset, P(H)]}, u_H).
\end{equation}
The function $\phi_H$ being  non-increasing and continuous,  for a fixed vector
$u_{[\emptyset, P(H)]}$, there exists  a unique solution $u_H=F_H(u_{[\emptyset, P(H)]})\geq 0$ to the above equation. 
Furthermore, the function $ u_{[\emptyset, P(H)]}\to F_H(u_{[\emptyset, P(H)]})$ is
continuous and non-increasing. For such an $H$, for $H'=P(H)$, Relation~\eqref{eq:leave} can then be written as 
\[
u_{H'}=\phi_{H'}(u_{[\emptyset, {P(H')}]},u_{H'}) + \sum_{G\in {\cal T},   {H'}\vdash G}  F_G(u_{[\emptyset,P(H')]},u_{H'}).
\]
Since $\phi_{H'}$ and $F_G$, with~$G$ a leaf, are non-increasing and continuous, there exists a
unique solution $u_{H'}=F_{H'}(u_{[\emptyset, {P(H')}]})\geq 0$ and the function
\[
u_{[\emptyset, P(H')]}\to F_{H'}(u_{[\emptyset, P(H')]}),
\] 
is continuous and non-increasing. By induction (by decreasing level of nodes), one obtains that a family of continuous, non-increasing functions
$F_G$, $G\in{\cal T}$, $G\not=\emptyset$, exists, such that, for a fixed vector $u_{[\emptyset, {P(G)}]}$, $u_G=F_{G}(u_{[\emptyset, {P(G)}]})$ is the unique solution  of 
\begin{equation*}
u_G=\phi_{G}(u_{[\emptyset, {P(G)}]},u_G) 
+ \sum_{G'\in {\cal T},   {G}\vdash G'} F_{G'}(u_{[\emptyset,P(G)]},u_G).
\end{equation*}
%\begin{multline*} 
%F_{G}(u_{[\emptyset, {P(G)}]})=\phi_{G}(u_{[\emptyset, {P(G)}]},F_{G}(u_{[\emptyset, {P(G)}]})) \\
%+ \sum_{G'\in {\cal T},   {G}\vdash G'} F_{G'}(u_{[\emptyset,P(G)]},F_{G}(u_{[\emptyset, {P(G)}]})).
%\end{multline*}
Equation~\eqref{eq:leave}  at the root then writes 
\[
u_\emptyset= \phi_\emptyset(u_\emptyset)+\sum_{G\in {\cal T}, \emptyset\vdash G}F_G(u_\emptyset),
\]
and this equation has a unique solution $\bar{u}_\emptyset$. Now, one defines recursively (by increasing level of nodes) 
\[
\bar{u}_G=F_G(\bar{u}_{[\emptyset,P(G)]}), \ \  G\in{\cal T}.
\]
Then clearly  $(\bar{u}_G, G\in{\cal T})$ satisfies Relation~\eqref{eq:leave} and
is the unique solution.
\end{proof}

\section{Linear topologies}\label{secLinear}
In this section we consider a linear network with $J$ nodes and $K=J+1$ classes of connections. Class $j$ connections,  $1\leq j\leq J$, use node~$j$ only, while class~0 connections use all $J$ nodes. Assume $A_{jk}=1$ if class $k$
uses node $j$, and 0 otherwise. 
Equation~(\ref{fp}) is in this case
\begin{equation}
\label{eq:lin}
u_j =\phi_{0}(u)+ \phi_j(u_j), \ 1\leq j\leq J,  
\end{equation}
with $u=(u_1,\ldots, u_J)$.

\begin{figure}[ht]
\begin{center}
\scalebox{.3}{\includegraphics{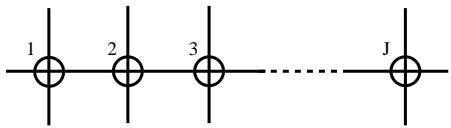}}
\end{center}
\caption{A linear network with $J$ nodes and $K=J+1$ classes of connections}
\end{figure}

\begin{proposition}
If  the functions, $\beta_k,  0\leq k\leq J$,  are  continuous and  non-decreasing, then there exists a unique solution for the 
fixed point equation~\eqref{fp}. 
\end{proposition}

\begin{proof}
Let $\bar \phi_{j}(x)=x- \phi_j(x)$, $x\in\mathbb{R}$,  which is continuous and non-decreasing.
Hence,~(\ref{eq:lin}) can be rewritten as
\begin{equation}
\label{eq:uj}
u_j=\bar \phi_j^{-1}\left(\phi_{0}(u)\right)=\bar \phi_j^{-1}\left(\frac{\alpha_{0}}{\sqrt{\beta_{0}(u)}}\right), \ 1\leq j\leq J.
\end{equation}
for some constant $\alpha_0$, see Equation~\eqref{Irancy}.
In addition, define the function
$\psi_j(x)=\bar\phi_j^{-1}\left({\alpha_{0}}/{\sqrt{x}}\right)$,  $x\in\mathbb{R}$,
which is continuous and  non-increasing. 
>From (\ref{eq:uj}) we obtain the relation 
$$
\beta_{0}(u) = \beta_{0}\left(\psi_1 (\beta_{0}(u) ),\ldots, \psi_J(\beta_{0}(u)) \right).
$$
Since $\beta_{0}$ is non-decreasing and $\psi_j$ is non-increasing, the fixed point
equation $\beta = \beta_{0}(\psi_1(\beta),\ldots,\psi_J(\beta))$ has a unique solution
$\beta^*\geq 0$. Hence, the Relation~(\ref{eq:uj}) has a unique fixed point, which is
given by $u^*_j=\bar \phi_j^{-1}\left({\alpha_{0}}/{\sqrt{\beta^*}}\right)$. 
\end{proof}

\section{Ring topologies}\label{Torus}
In this section, the topology of the network is based on a ring.  Several situations are
considered for the paths of the connections.
\subsection*{Routes with two consecutive nodes}
It is assumed that there are $J$ nodes and $K=J$ classes of connections and class
$j\in\{1,\ldots, J\}$ uses two nodes: node $j$ and $j+1$. Assume $A_{jk}=1$ if class $k$
uses node $j$, and 0 otherwise. 
%With the same notations as in Theorem~\ref{thinv}, $\phi_j({x})=\phi_j(x_j,x_{j+1})=\alpha_j/\sqrt{\beta_j(x_j,x_{j+1})}$.
Equation~\eqref{fp} is in this case
\begin{equation}
\label{eq:torus_u_new}
u_j=   \phi_{j-1}( u_{j-1}, u_j)  + \phi_{j}(u_{j}, u_{j+1}),  \ j=1,\ldots, J.
\end{equation}
For  $y_j= \phi_j(u_j,u_{j+1})$, the above equation can be rewritten as follows
\begin{equation}
\label{eq:lin_y}
y_j= \phi_j(y_{j-1}+ y_j, y_j+ y_{j+1}), \ j=1,\ldots, J.
\end{equation}

\begin{figure}[ht]
\begin{center}
\scalebox{.3}{\includegraphics{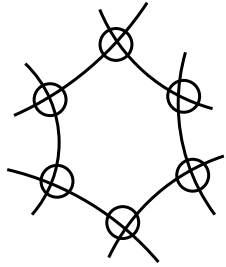}}
\end{center}
\caption{Routes with two consecutive nodes}
\end{figure}

\begin{proposition}
If the functions $\beta_k$, $1\leq k\leq K$, are  continuous, non-decreasing and satisfy the assumptions of
Proposition~\ref{prop:fp_exist}, then there exists a unique solution for the fixed point
equation~\eqref{fp}.
\end{proposition}
\begin{proof}
>From Proposition~\ref{prop:fp_exist} we have that Equation~\eqref{eq:lin_y} has at least one
fixed point solution. Let $x=(x_j:j=1,2,\ldots, J)$ and
$y=(y_j:j=1,2,\ldots,J)$ both be fixed points.

If the relation $y_j<x_j$ holds for all $j=1,\ldots,J$, then the inequality 
\[
\phi_j( y_j+ y_{j-1}, y_{j+1}+y_j)=y_j< x_j=\phi_j( x_j+ x_{j-1}, x_{j+1}+x_j),
\]
and the fact that the function $\phi_j$ is non-increasing, give directly a  contradiction.

Consequently, possibly up to an exchange of $x$ and $y$, one can assume that there exists  $m\in
\{1,\ldots, J\}$ such that $y_m\leq x_m$ and $y_{m+1}\geq x_{m+1}$.
Define $c_j=x_j-y_j$ and $d_{j}=y_{j}-x_{j}$.  Hence, $c_m\geq 0$ and $d_{m+1}\geq 0$.
Without loss of generality, it can be  assumed that the classes are
ordered such that $d_{m-1}\leq d_{m+1}$.
Since the function $\phi_m$ is non-increasing, and 
\[
\phi_m( 
y_m +y_{m-1}, y_m+ y_{m+1} )=y_m \leq x_m= \phi_m( x_m +x_{m-1}, x_m+ x_{m+1}),
\]
we have that either
\[
    y_m +y_{m-1}
    \geq
   x_m +x_{m-1} \text{ and/or }     y_m +y_{m+1}
    \geq
   x_m +x_{m+1},
\]
 i.e., $\ d_{m-1}\geq c_m$ and/or $\ d_{m+1} \geq c_m$.  Because $d_{m-1}\leq d_{m+1},$
then, necessarily, $ d_{m+1} \geq  c_m\geq 0$. Hence 
\begin{multline*}
\phi_{m+1}( y_{m+1} +y_{m}, 
y_{m+1}+y_{m+2})=y_{m+1}\\ \geq x_{m+1}= \phi_{m+1}( x_{m+1} +x_{m}, x_{m+1}+x_{m+2}). 
\end{multline*}
Since $\phi_{m+1}$ is non-increasing, 
one has  $y_{m+1} + y_{m+2} \leq  x_{m+1} + x_{m+2}$ and consequently  $d_{m+1} \leq c_{m+2}$.

>From $0\leq c_m\leq d_{m+1} \leq  c_{m+2}$, we obtain  $x_{m+2}\geq y_{m+2}$, which, using the same steps as before, implies $c_{m+2}\leq d_{m+3}$.
In particular, by  induction it can be concluded that 
$$c_j\leq d_{j+1} \leq c_{j+2} \leq d_{j+3}, \ \text{for all } j=1,\ldots, J,$$
where the indices $j+1, j+2$, and $j+3$ are considered as modulo $J$. 
%$c_j\leq  c_{j+2}$ and $d_{j-1}\leq d_{j+1}$ for $j=m+2i$ and $i=0, 1,\ldots$.
This implies that $c_j=d_j=c$, for all $j=1,\ldots,L$, and hence  $y_j+y_{j-1}=x_j+x_{j-1}$, i.e.,     $$y_j= \phi_j(y_j+y_{j-1}, y_j+y_{j+1}) = 
\phi_j( x_j+x_{j-1}, x_j+x_{j+1})=x_j,$$ for $j=1,\ldots,J$.
We can conclude that the fixed point is unique.
\end{proof}

%%%Contraction
The rest of this part will be devoted to a contraction argument that can be used to get a
unique solution to the fixed point equation.

\begin{proposition}\label{Bell}
If the functions $\beta_k$, $1\leq k\leq K$, are Lipschitz, continuous differentiable, and non-decreasing, then there exists a unique solution for the fixed point
equation~\eqref{fp}.
\end{proposition}
\begin{proof}
The proof consists in showing that~(\ref{eq:lin_y}) has a unique solution. 
By the Implicit function theorem, there exists a unique 
$x_{j}(y_{j-1}, y_{j+1})$ such that, 
\begin{equation} \label{eq:torus_func}
x_{j}(y_{j-1}, y_{j+1})= \phi_j( y_{j-1} + x_{j}(y_{j-1}, y_{j+1}) ,  x_{j}(y_{j-1}, y_{j+1}) + y_{j+1}),
\end{equation}
and this function $(y_{j-1}, y_{j+1})\to x_j(y_{j-1}, y_{j+1})$ is positive and continuous
differentiable. Taking the partial derivative to
$y_{j-1}$ on both sides of this identity, one gets that 
\begin{multline*}
\frac{\partial x_{j}(y_{j-1}, y_{j+1})}{\partial y_{j-1}}
= \left. \frac{\partial \phi_j(s_1,s_2)}{\partial s_1}\right|_{s=s(y)} \times \left(1 +   \frac{\partial x_{j}(y_{j-1}, y_{j+1})}{\partial y_{j-1}}\right)
\\+
\left.  \frac{\partial \phi_j(s_1,s_2)}{\partial s_2}\right|_{s=s(y)} \times \frac{\partial x_{j}(y_{j-1}, y_{j+1})}{\partial y_{j-1}},
\end{multline*}
with $s(y)=(  y_{j-1} +  x_{j}(y_{j-1}, y_{j+1}) , x_{j}(y_{j-1}, y_{j+1}) +  y_{j+1})$. Hence,
\[
\frac{\partial x_{j}(y_{j-1}, y_{j+1})}{\partial y_{j-1}}
= \left.\left[\frac{\partial \phi_j(s_1,s_2)}{\partial s_1}
\left/\left(1- \frac{\partial \phi_j(s_1,s_2)}{\partial s_1} -
 \frac{\partial \phi_j(s_1,s_2)}{\partial s_2}\right)\right.\right]\right|_{s=s(y)}  {\leq} 0.
\]
A similar expression holds for 
$
{\partial x_{j}(y_{j-1}, y_{j+1})}/{\partial y_{j+1}}\leq 0,
$
and one can conclude that 
\begin{align}
\label{eq:torus_sum_der}
&\left|\frac{\partial x_{j}(y_{j-1}, y_{j+1})}{\partial y_{j-1}}\right|+\left|\frac{\partial x_{j}(y_{j-1}, y_{j+1})}{\partial y_{j+1}}\right|
\\&= -\left[\left(\frac{\partial \phi_j(s_1,s_2)}{\partial s_1} +
\frac{\partial \phi_j(s_1,s_2)}{\partial s_2}\right) 
\left/
\left(1- \frac{\partial \phi_j(s_1,s_2)}{\partial s_1} - 
\frac{\partial \phi_j(s_1,s_2)}{\partial s_2}
\right)
\right]
\right|_{s=s(y)}.\notag
\end{align}
If $x_j(0,y_{j+1})=0$, then by Relation~\eqref{eq:torus_func}, one
gets that, for some constant $\alpha_j$, see Equation~\eqref{Irancy},
\[
0=\phi_j(0, y_{j+1})=\alpha_j/\sqrt{\beta_j(0,y_{j+1})},
\]
which holds only if $y_{j+1}=\infty$. Hence,  $x_j(0,y_{j+1})>0$.
Since $x_j(y_{j-1},y_{j+1})$ is continuous and positive, and  $x_j(0,y_{j+1})>0$, one
obtains that there exists an $M_j^->0$ such that $y_{j-1} +  x_{j}(y_{j-1},
y_{j+1})>M_j^-$ for all $y_{j-1}, y_{j+1}\geq 0$. 
Similarly, there exists an $M_j^+>0$ such that $x_{j}(y_{j-1}, y_{j+1})+y_{j+1}>M_j^+$.
This gives the following upper bound,
\begin{multline*}
-\left. \frac{\partial \phi_j(s_1,s_2)}{\partial s_i}\right|_{s=s(y)}=\frac{\alpha_j}{2} 
{\left.\frac{\partial \beta_j (s_1,s_2)}{\partial
    s_i}\right|_{s=s(y)}}{\beta_j(s(y))^{-3/2}}\\ \leq
\frac{\alpha_j}{2}\frac{L}{(\beta_j(M_j^-,M_j^+))^{3/2}},
\end{multline*}
where we used that $\beta_j$ is non-decreasing, Lipschitz continuous (with constant $L$) and  differentiable.
>From Equation~\eqref{eq:torus_sum_der}) one now obtains that there exists a constant $0<C<1$ such that
\[
\left|\frac{\partial x_{j}(y_{j-1}, y_{j+1})}{\partial
  y_{j-1}}\right|+\left|\frac{\partial x_{j}(y_{j-1}, y_{j+1})}{\partial
  y_{j+1}}\right|<C.
\]
Hence, the mapping $T:\mathbb{R}_+^J\to \mathbb{R}_+^J$ with $T(y)=
(x_j(y_{j-1},y_{j+1})$ for $ j=1,\ldots,J)$ is a contraction, and has a unique fixed point
$(y^*_j)$, i.e.,  Equation~\eqref{eq:lin_y} has a unique solution.
\end{proof}

\subsection*{Routes with one node or two consecutive nodes}
Consider now a ring with $J$ nodes and  $K=2J$ classes.   Class $j$ uses two nodes: nodes $j$ and $j+1$, $j=1,\ldots, J$. Class $0j$ uses one node:
node  $j$, $j=1,\ldots, J$. 
We assume that  $A_{jk}=1$ if
and only if class~$k$ uses node~$j$, and zero otherwise.

\begin{figure}[ht]
\begin{center}
\scalebox{.3}{\includegraphics{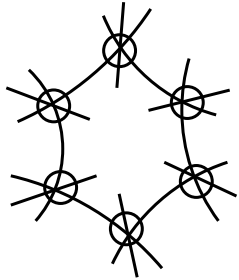}}
\end{center}
\caption{Routes with one node or two consecutive nodes}
\end{figure}
We focus on functions $\beta_k$ that satisfy~(\ref{eq:lossr}).  Equation~\eqref{fp} is in this context 
\[
u_j=\phi_{0j}(u_j)+\phi_{j-1}(u_{j-1}+u_j) +\phi_j(u_j+u_{j+1}),\quad 1\leq j\leq J.
\]
For $y_j=\phi_j(u_{j}+u_{j+1})$ and $y_{0j}= \phi_{0j}(u_j)$, $j=1\ldots,J$, the above
equation can be rewritten as follows, for $j=1,\ldots, J$,
\begin{equation}\label{eq:torusB_oj}
\begin{cases}
y_j&=\phi_j(y_{j-1}+2y_j + y_{j+1}+y_{0j}+ y_{0j+1} ),\\
y_{0j}&=\phi_{0j}( y_{0j}+y_{j-1}+y_{j} ).
\end{cases}
\end{equation}
\begin{proposition}
If the functions $\beta_k$, $1\leq k\leq J$,  $01\leq k\leq 0J$, are Lipschitz, continuously differentiable, 
non-decreasing, and  satisfy~(\ref{eq:lossr}),  
then there exists a unique solution for the fixed point equation~\eqref{fp}.
\end{proposition}
\begin{proof}
By the Implicit function theorem, for each $j$,  there exists a unique 
$x_{0j}(t)$  satisfying the relation $x_{0j}(t)=\phi_{0j}(x_{0j}(t)+ t)$, and this function is  non-increasing and continuous differentiable.  
One  now has to solve the equation
\begin{equation}
\label{eq:torusB_yj}
y_j= \phi_j(y_{j-1}+2y_j +y_{j+1}+x_{0j}(y_{j-1}+y_j)+x_{0,j+1}(y_{j}+y_{j+1})) .
\end{equation}
>From the fact that $-1\leq x_{0j}'(t)\leq 0$, it can be easily checked that the right-hand side of Equation~(\ref{eq:torusB_yj}) is non-increasing in $y_j$.
%
\iffalse
Define the function 
\[
g_j(y_j)=\phi_j(y_{j-1}+2y_j +y_{j+1}+x_{0j}(y_j+y_{j-1})+x_{0,j+1}(y_{j+1}+y_j)). 
\]
Since  $-1<x_{0j}'(t)<0$, the   derivative  of this function is 
\[
g_j'(y_j)=(2+ x_{0j}'(y_j +y_{j-1})+x_{0,j+1}'(y_{j+1}+y_j ))\cdot
\phi_j'(d_j(y))\leq 0,
\]
with
\[
d_j(y)=y_{j-1}+2y_j +y_{j+1}+x_{0j}(y_j+y_{j-1})+x_{0,j+1}(y_{j+1}+y_{j}),
\]
so $g_j$ is non-increasing in $y_j$.  
\fi
%
Hence,  there exists a unique   $x_j(y_{j-1},y_{j+1})$ such that $y_j=x_j(y_{j-1},y_{j+1})$
satisfies Equation~\eqref{eq:torusB_yj}, 
%, i.e.  $x_j(y_{j-1},y_{j+1})= \phi_j(f_j(y))$,  for all $j=1,\ldots, L$, with  
and this function $(y_{j-1}, y_{j+1})\to x_j(y_{j-1}, y_{j+1})$ is positive and continuous
differentiable (by the Implicit function theorem).
In particular, $x_j(y_{j-1},y_{j+1})= \phi_j(f_j(y))$,  for all $j=1,\ldots, L$, with  
\begin{multline*} f_j(y)=y_{j-1}+2x_j(y_{j-1},y_{j+1}) +y_{j+1}\\+x_{0j}(y_{j-1}+x_j(y_{j-1},y_{j+1}))+x_{0,j+1}(x_j(y_{j-1},y_{j+1})+y_{j+1}). \end{multline*}
>From this one can derive that 
\begin{multline*}
\frac{\partial x_j(y_{j-1},y_{j+1})}{\partial y_{j-1}}= 
\left.\phi_j'(f_j(y))\left[1+x_{0j}'(y_{j-1}+x_j(y_{j-1},y_{j+1}))\right]\right/\\
\left[
\rule{0mm}{4mm}1- 
\phi_j'(f_j(y))(2+x_{0j}'(y_{j-1}+x_j(y_{j-1},y_{j+1}))\right.\\\left.\rule{0mm}{4mm} +x_{0,j+1}'(x_j(y_{j-1},y_{j+1})+y_{j+1})) \right] \leq 0,
\end{multline*}
and a similar expression holds for 
$
{\partial x_j(y_{j-1},y_{j+1})}/{\partial y_{j+1}}\leq 0.
$
As in the proof of Proposition~\ref{Bell}, an upper bound on $-\phi_j'(f_j(y))$ can be obtained. This implies that there exists a constant $0< C <1$ such that 
\[
\left|\frac{\partial
  x_j(y_{j-1},y_{j+1})}{\partial y_{j-1}}\right|+\left|\frac{\partial
  x_j(y_{j-1},y_{j+1})}{\partial y_{j+1}}\right| 
< C,\quad j=1,\ldots , J.
\]
Hence, the mapping $T:\mathbb{R}_+^J\to \mathbb{R}_+^J$ with $T(y_1, \ldots,
y_J)=(x_j(y_{j-1},y_{j+1}), j=1,\ldots, J)$ is a contraction, and  has a unique fixed point $(y_j^*)$. 
One concludes that there exists a unique solution $y_j^*$ and $y_{0j}^*=x_{0j}(y_{j-1}^*+y_j^*)$, $j=1,\ldots,J$,  of   
\eqref{eq:torusB_oj}.
\end{proof}

\subsection*{Routes with two consecutive nodes and a complete route}
Consider a ring with $J$ nodes and  $K=J+1$ classes.  Class $1\leq j\leq J$ uses two nodes: node  $j$ and $j+1$ and class $0$
uses all nodes $1,\ldots, J$.  

\begin{figure}[ht]
\begin{center}
\scalebox{.3}{\includegraphics{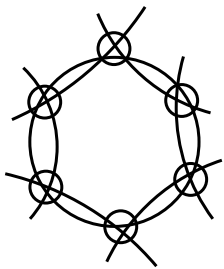}}
\end{center}
\caption{Routes with two consecutive nodes and a complete route}
\end{figure}
We focus on functions $\beta_k$ that satisfy~(\ref{eq:lossr}).
Equation~\eqref{fp} is in this context
\[
u_j= \phi_0\left(u_1+\cdots+u_J\right) +  \phi_{j-1}( u_{j-1} +u_j)  +\phi_{j}(u_{j}+ u_{j+1}).
\]
For  $y_j= \phi_j(u_j+u_{j+1})$ and  $y_0=\phi_0(u_1+u_2+\cdots+u_J),$   the above
equation can be rewritten as follows 
\begin{equation} \label{eq:torusA_yj}
\begin{cases}
y_j&=\phi_j(y_{j-1}+2y_j+ y_{j+1}+2y_0), \quad j=1,\ldots, J,\\
y_0&=\phi_0( Jy_0 + 2\sum_{j=1}^J y_j ).
\end{cases}
\end{equation}

\begin{proposition} 
If the functions $\beta_k$,  $0\leq k\leq J$,  are  continuous, non-decreasing, satisfy~(\ref{eq:lossr}), and satisfy the
assumptions of Proposition~\ref{prop:fp_exist}, then there exists a unique solution for the 
fixed point equation~\eqref{fp}. 
\end{proposition}
\begin{proof}
>From Proposition~\ref{prop:fp_exist} we have that  Equation~\eqref{eq:torusA_yj} has at least one fixed point solution. Let $x=(x_j:j=0,1,\ldots, J)$ and
$y=(y_j:j=0,1,\ldots,J)$ both be fixed points. 

If the relation $y_j<x_j$ holds for all $j=1,\ldots,J$, then 
\[
\phi_j(2y_j+ y_{j-1}+y_{j+1}+2y_0)=y_j<
x_j=\phi_j(2x_j+ x_{j-1}+x_{j+1}+2x_0).
\]
Since the function $\phi_j$ is non-increasing, one gets that
\[
2y_j+ y_{j-1}+y_{j+1}+2y_0>2x_j+ x_{j-1}+x_{j+1}+2x_0.
\]
Summing over all $j=1,\ldots, L$,  we obtain  
\[
4(y_1+\cdots+y_J)+ 2Jy_0>4 (x_1+\cdots+x_J)+ 2Jx_0,
\]
which implies that $x_0<y_0$. However, 
\[
y_0=\phi_0\left(2(y_1+\cdots+y_J) + Jy_0\right) \leq \phi_0\left(2(x_1+\cdots+x_J)+ Jx_0\right) =x_0, 
\]
hence,  we obtain a contradiction. 

We can conclude that there  is an  $m\in \{1,\ldots,  J\}$ such that  $y_m\leq x_m$  and $y_{m+1}\geq
x_{m+1}$. 
To show that $x=y$, one proceeds along similar lines as in the proof of Proposition~\ref{Bell}. 
\end{proof}

\providecommand{\bysame}{\leavevmode\hbox to3em{\hrulefill}\thinspace}
\providecommand{\MR}{\relax\ifhmode\unskip\space\fi MR }
\providecommand{\MRhref}[2]{%
  \href{http://www.ams.org/mathscinet-getitem?mr=#1}{#2}
}
\providecommand{\href}[2]{#2}

\end{document}